\documentclass[journal,twoside,web]{ieeecolor}

\usepackage{cite}
\usepackage{amsmath,amssymb,amsfonts}
\usepackage{mathtools}
\usepackage{graphicx}
\graphicspath{{photos/}}
\usepackage{color}
\usepackage{enumerate}   
\usepackage{algpseudocode}
\usepackage{algorithm}
\usepackage{siunitx}
\usepackage{bbold}
\usepackage{multicol}
\usepackage{booktabs}
 \usepackage{multirow} 
\usepackage{hyperref}
\usepackage{textcomp}
\usepackage{subfigure}
\usepackage{stfloats}
\usepackage{epsfig}
\usepackage{lcsys}

\definecolor{Matlab_blue}{rgb}{0, 0.4470, 0.7410} 
\definecolor{Matlab_orange}{rgb}{0.8500, 0.3250, 0.0980} 
\definecolor{Matlab_yellow}{rgb}{0.9290, 0.6940, 0.1250} 
\definecolor{green_dark}{rgb}{0.0, 0.5, 0.0}

\newcommand{\Nset}{\mathbb{N}}
\newcommand{\Rset}{\mathbb{R}}

\newcommand{\Zset}{\mathbb{Z}}
\newcommand{\Uset}{\mathbb{U}}

\newcommand{\bu}{\mathbf{u}}

\newcommand{\bz}{\mathbf{z}}

\newcommand{\col}{\operatorname{col}}

\newcommand{\figurename}[1]{Fig.#1}

\newtheorem{theorem}{Theorem}[section]
\newtheorem{problem}[theorem]{Problem}
\newtheorem{lemma}[theorem]{Lemma}

\newtheorem{definition}[theorem]{Definition}
\newtheorem{remark}[theorem]{Remark}
\newtheorem{assumption}[theorem]{Assumption}

\begin{document}

\def\BibTeX{{\rm B\kern-.05em{\sc i\kern-.025em b}\kern-.08em
    T\kern-.1667em\lower.7ex\hbox{E}\kern-.125emX}}
\markboth{\journalname, VOL. XX, NO. XX, XXXX 2017}
{Author \MakeLowercase{\textit{et al.}}: Preparation of Papers for IEEE Control Systems Letters (August 2022)}

\title{Kernelized offset--free data--driven predictive control for nonlinear systems}

\author{
Thomas de Jong$^{1}$,
Mircea Lazar$^{1}$
\thanks{$^{1}$~Department of Electrical Engineering,
        Eindhoven University of Technology, 5612 AZ Eindhoven, The Netherlands. E-mails of the authors: \texttt{t.o.d.jong@tue.nl, m.lazar@tue.nl}.}
}

\maketitle
\thispagestyle{empty}

\begin{abstract}
This paper presents a kernelized offset-free data-driven predictive control scheme for nonlinear systems. Traditional model-based and data-driven predictive controllers often struggle with inaccurate predictors or persistent disturbances, especially in the case of nonlinear dynamics, leading to tracking offsets and stability issues. To overcome these limitations, we employ kernel methods to parameterize the nonlinear terms of a velocity model, preserving its structure and efficiently learning unknown parameters through a least squares approach. This results in a offset-free data-driven predictive control scheme formulated as a nonlinear program, but solvable via sequential quadratic programming. We provide a framework for analyzing recursive feasibility and stability of the developed method and we demonstrate its effectiveness through simulations on a nonlinear benchmark example.
\end{abstract}


\section{Introduction}
\label{sec:introduction}
\IEEEPARstart{M}{odel} predictive control (MPC) is a control strategy that optimizes predicted system behavior, while ensuring system constraints. Model inaccuracies and/or persistent disturbances require offset-free design methods for MPC, which include incremental/velocity state-space models or augmented disturbance models, see, e.g., the overview \cite{PANNOCCHIA2015}. For nonlinear MPC, velocity state-space models were developed in \cite[Chapter~4]{cisneros2021quasi} via the multivariable mean value theorem. 

Recently, data-driven approaches to linear predictive control, such as subspace predictive control (SPC) \cite{FAVOREEL19994004} and data-enabled predictive control (DeePC) \cite{coulson2019data} have become popular due to a rising interest in data-driven control. In \cite{lazar2022offset}, it was shown that  using incremental inputs in the SPC or DeePC design yields offset-free data-driven predictive controllers with similar performance as offset-free linear MPC. A first result in this direction for nonlinear systems was presented in \cite{cisneros2020}, where finite-dimensional Koopman operators were used to approximate the nonlinear dynamics. A velocity-form prediction model was subsequently derived therein by applying the approach in \cite[Chapter~4]{cisneros2021quasi} to the Koopman approximate model. Recently, in \cite{lazar2024basis} it was shown that general methods to construct basis functions, such as neural networks or kernel methods, can be used to design data-driven predictive controllers for nonlinear systems, without knowledge of the system dynamics. 

Among different methods to construct basis functions, kernel methods offer a well-posed formulation in reproducing kernel Hilbert spaces (RKHSs) \cite{scholkopf2001generalized}, providing a tractable alternative to neural networks for system identification and data-driven analysis/control, see, e.g., \cite{Kernels_control_Claudio, Timm_Ker_Lemma}. Within the data-driven predictive control field, kernels were originally used in \cite{Maddalena_KPC} to learn multi-step predictors for nonlinear systems, in the spirit of the SPC formulation for linear systems. More recently, in \cite{huang2023robust}, kernels were used to implicitly parameterize multi-step predictors in a kernelized DeePC formulation for nonlinear systems. These formulations, however, do not address the offset-free design problem. In this regard, data-driven velocity-form state-space models for nonlinear systems were originally introduced in \cite{10384139}, based on the fundamental theorem of calculus and linear-parameter-varying (LPV) embeddings. Therein, a  parametrization of the velocity-form state-space model using basis or kernel functions was defined, but formulation of the associated learning problem was not explicitly addressed. Very recently, multi-step data-driven velocity-form models for nonlinear systems were presented in \cite{verhoek2024kernel}, based on input-output LPV representations and kernel functions. 

In this paper we present a novel kernelized offset-free DPC scheme for nonlinear systems with some desirable features, as follows. Firstly,  we consider the velocity state-space model as in \cite[Chapter 4]{cisneros2021quasi} and we employ kernel functions in a structured way, i.e., by representing each nonlinear function/gradient in the velocity model.  We show that preserving the structure of the velocity model in \cite[Chapter 4]{cisneros2021quasi} leads to a tractable least squares problem for learning the unknown coefficients and yields a kernelized velocity model of the same dimension as the original velocity model. This is more efficient and scalable with the data size compared to the state-of-the-art, e.g., the approach in \cite{cisneros2020} requires solving a \emph{nonlinear} least squares problem to estimate the parameters of the Koopman observables and computing partial derivatives of the resulting Koopman embedding. Secondly, we formulate the kernelized offset-free DPC problem as a parameterized nonlinear program, which can be solved via a sequence of quadratic programs (QPs), as done in \cite{cisneros2021quasi} for an analytic velocity form model. Lastly, we present terminal cost and constraint set conditions for recursive feasibility and stability of nonlinear velocity form predictive control, which apply to both analytic and data-driven velocity form representations. This is novel compared to the state-of-the-art \cite[Chapter 4]{cisneros2021quasi}, which only establishes output convergence under a conservative, terminal equality output and incremental state constraint.


\paragraph*{Notation and basic definitions} Let $\mathbb{R}$, $\mathbb{R}_{+}$ and $\mathbb{N}$ denote the set of real numbers, the set of non-negative reals, and the set of non-negative integers. 
For every $c \in \mathbb{R}$ and $\Pi \subseteq \mathbb{R}$ define $\Pi_{\geq c} := \{k \in \Pi | k \geq c \}$. By $\text{Co}(a,b)$ we denote the convex hull of vectors $a,b$. Given a signal $v \in \Rset^{n_v}$, a starting
time $k \geq 0$ and and ending instant $j \geq k + 1$, we define
$\Bar{v}(k, j) = \col (v(k), . . . , v(k + j - 1))$. For a matrix $A \in \Rset^{n\times m}$ we denote by $A^\dag$ the Moore–Penrose inverse of $A$. For a square matrix $A$, we denote by $A\succ 0$ that $A$ is positive definite, and by $A\succeq 0$ that $A$ is positive semi-definite. The identity matrix of dimension $n$ is denoted by $I_n$ and $\mathbf{0}_{n\times m}$ is used to denote a matrix of zeros with $n$ rows and $m$ columns; we omit the subscript when the dimensions are clear from the context. 
\section{Preliminaries \& Problem Statement}\label{sec:preliminaries_and_problem_statement}
We consider nonlinear MIMO systems with inputs $u\in\mathbb{R}^m$, measured states and outputs $x\in\mathbb{R}^n$, $y\in\mathbb{R}^p$, affected by piece-wise constant additive disturbances $d\in\mathbb{R}^{n_d}$, i.e.,
\begin{equation}\label{eqn:state_space}
    \begin{aligned}
         x_{k+1} &= f(x_{k},u_{k}) + E d_k, \quad k\in\mathbb{N}, \\
        y_{k} &= h(x_{k}).
    \end{aligned}
\end{equation}
The functions $f:\mathbb{R}^{n} \times \mathbb{R}^{m} \rightarrow \mathbb{R}^{n}$ and $h:\mathbb{R}^{n} \rightarrow \mathbb{R}^{p}$ are assumed to be \emph{unknown}. Next, we introduce the incremental state and input $\Delta x_k = x_k - x_{k-1}$, $\Delta u_k = u_k - u_{k-1}$ and the extended state $z_k := \col(y_{k-1},\Delta x_k)$. Define the incremental/velocity dynamics as in \cite[Chapter 4]{cisneros2021quasi}:
\begin{equation}\label{eqn:velocity}
    \begin{aligned}
          z_{k+1} &= \underbrace{\begin{bmatrix}
        I & \nabla_x h(\Tilde{x}_{k}) \\
        0 & \nabla_x f(\Tilde{x}_{k},\Tilde{u}_{k})
    \end{bmatrix}}_{A(\Tilde{x}_k,\Tilde{u}_k)} z_{k} + \underbrace{\begin{bmatrix}
        0 \\
       \nabla_u f(\Tilde{x}_{k},\Tilde{u}_{k}) 
    \end{bmatrix}}_{B(\Tilde{x}_k,\Tilde{u}_k)} \Delta u_k, \\
    y_k  &=  \underbrace{\begin{bmatrix}
        I & \nabla_x h(\Tilde{x}_{k}) 
    \end{bmatrix}}_{C(\Tilde{x}_k)} z_{k},
    \end{aligned}
\end{equation}
where $\Tilde{x}_k\in\text{Co}(x_{k},x_{k+1})$ and $\Tilde{u}_k\in\text{Co}(u_{k},u_{k+1})$. The velocity-form representation \eqref{eqn:velocity} is exact as it gives an exact dynamic equation in a space tangent to the original state space (the velocity space) \cite{cisneros2021quasi}. Note that the elements of the matrices in \eqref{eqn:velocity} are in general nonlinear functions. Since analytic expressions for $\Tilde{x}_k$ and $\Tilde{u}_k$ cannot be determined, \cite[Chapter~4]{cisneros2021quasi} employs an approximate velocity-form model based on the assumption that $\Tilde{x}_{k} \approx x_k$ and $\Tilde{u}_{k} \approx u_k$ in \eqref{eqn:velocity}. Next we introduce the following instrumental definitions and results on reproducing kernel Hilbert spaces, see e.g., \cite{fasshauer2011positive,paulsen2016introduction}.
\begin{definition}
Given a set $\mathcal{X}$, we will say that $\mathcal{H}$ is a reproducing kernel Hilbert space (RKHS) on $\mathcal{X}$ over $\mathbb{R}$, provided that (\emph{i}) $\mathcal{H}$ is a Hilbert space of functions from $\mathcal{X}$ to $\mathbb{R}$ and (\emph{ii}) for every $y \in \mathcal{X}$, the linear evaluation functional, $E_y : \mathcal{H} \rightarrow \Rset$, defined by $E_y(f) = f(y)$, is bounded.
\end{definition}
\begin{definition}
A symmetric function $K: \mathcal{X} \times \mathcal{X} \rightarrow \Rset$ is called positive semidefinite kernel if, for any $h \in \Nset$
$$
\sum_{i=1}^h \sum_{j=1}^h \alpha_i \alpha_j K(x_i,x_j) \geq 0, \,\, \forall (x_k,\alpha_k)\in \left(\mathcal{X},\Rset \right), k\in [1:h].
$$
The kernel section of $K$ centered at $x$ is $K_x(\cdot) = K(x,\cdot)$, for all $x \in \mathcal{X}$.
\end{definition}

An RKHS is fully characterized by its reproducing kernel. To be specific, if a function $f : \mathcal{X} \rightarrow \mathbb{R}$ belongs to an RKHS $\mathcal{H}$ with kernel $K$, there exists a sequence $\alpha_i K_{x_i}(x)$, such that
$$
f(x) = \lim_{h\rightarrow\infty}\sum_{i=1}^{h}\alpha_i K_{x_i}(x)
$$
for some $\alpha_i \in \Rset, x_i \in \mathcal{X}$ , for all $i \in [1:h]$. Note that $h$ can possibly be a finite number. Next, suppose we are given pair-wise data points 
    $(x_1,y_1), \dots, (x_s,y_s)$, with each pair in $\mathcal{X} \times \Rset$. Then it holds that any $f\in\mathcal{H}$ \emph{minimizing the risk functional }
$$
c(x_1,y_1,f(x_1)),\dots,c(x_s,y_s,f(x_s)) + g(\|f\|_\mathcal{H})
$$
admits a representation  
\begin{equation}\label{eqn:representer_theorem}
    f(\cdot) = \sum_{i=1}^s \alpha_i K(\cdot,x_i),
\end{equation}
where $\|f\|_\mathcal{H}$ is the induced norm on $\mathcal{H}$, $g$ is a strictly monotonically increasing real-valued function on $[0, \infty)$ and $c :(\mathcal{X}\times \Rset^2)^s \rightarrow \Rset \cup \{\infty\}$ is an arbitrary cost function. This result is known as the representer theorem, see e.g., \cite{scholkopf2001generalized}. The case of multiple outputs will be addressed by employing a common, shared RKHS for all outputs. 

\subsection{Problem Statement}
To formulate the prototype offset-free nonlinear MPC problem, at time $k \in \Nset$, given $y_{k}$, $u_{k}$, $x_{k}$ and $N \in \Nset_{\geq 1}$, define:
\begin{align*}
    \mathbf{u}_k &:= \col(u_{0|k},\dots,u_{N-1|k}), \,
     \mathbf{y}_k := \col(y_{0|k},\dots,y_{N-1|k}), \\
    \mathbf{x}_k &:= \col(x_{0|k},\dots,x_{N-1|k}), \,
    \boldsymbol{\rho}_k := \col(\rho_{0|k},\dots,\rho_{N-1|k}),   \\
        \Delta \mathbf{u}_k &:= \col(\Delta u_{0|k},\dots,\Delta u_{N-1|k}), \, 
    \mathbf{z}_k := \col(z_{1|k},\dots, z_{N|k}). 
\end{align*}
where $\rho_{i|k} := \col(\Tilde{x}_{i|k},\Tilde{u}_{i|k})$. Next, consider the offset-free nonlinear MPC problem:

\begin{problem} (Velocity form offset-free nonlinear MPC) \label{prob:NMPC}
    \begin{subequations}
\label{eq:3:1}
\begin{align}
&\min_{\Delta \mathbf{u}_k,\boldsymbol{\rho}_k,\mathbf{z}_{k} }  \quad   l_N(z_{N|k},r) + \sum_{i=0}^{N-1} l(z_{i|k},r,\Delta u_{i|k})    \label{eq:3:1a} \\ 
&\text{subject to: } \nonumber  \\
&\mathbf{z}_k =  \Psi(\boldsymbol{\rho}_k) z_{0|k} + \Gamma(\boldsymbol{\rho}_k) \Delta \mathbf{u}_k \label{eq:3:1b}\\
 &\rho_{0|k} = \col(x_k,u_{0|k}), \quad z_{0|k} = \col(y_{k-1},\Delta x_k), \label{eq:3:1d}\\
 & (\mathbf{z}_{k}, \Delta \mathbf{u}_{k}) \in \Zset \times \Delta \Uset, \quad (z_{N|k}-r) \in \mathbb{Z}_T. \label{eq:3:1e}
\end{align}
\end{subequations}
\end{problem}

As a reference signal in Problem~\ref{prob:NMPC} we use $r=\col(y_r,\mathbf{0}_{n\times 1})$ where $y_r\in\mathbb{R}$ is a output reference, the terminal cost is $l_N(z_{N|k},r):=(z_{N|k}-r)^\top P(z_{N|k}-r)$, and 
the stage cost is $l(z_{i|k},r,\Delta u_{i|k}):=(z_{i|k}-r)^\top Q(z_{i|k}-r)+ \Delta u_{i|k}^\top R \Delta u_{i|k}$. We assume that $P \succ 0$, $Q \succ 0$ and $R \succ 0$. For an expression of $\Psi$ and $\Gamma$ we refer to Section~IV due to space limits. Since we assume that the functions $f, h$ are unknown, in order to build the velocity form model \eqref{eqn:velocity} and the prediction matrices $\Psi, \Gamma$, we need a data-driven representation of the unknown functions in \eqref{eqn:velocity}, which is formally stated next.
\begin{problem}\label{Problem:1}
Given a finite set of noiseless  state-input-output data $\{x_i,u_i,y_i\}_{i\in[0,s]}$, where $s\in\Nset_{\geq 1}$ denotes the number of samples, the objective is to construct a velocity model of the form \eqref{eqn:velocity} directly from data, i.e., without knowledge of the underlying system dynamics and corresponding gradients. Additionally, the aim is to develop a computationally efficient predictive control algorithm based on the developed data-driven velocity model and to analyze conditions under which recursive feasibility and closed-loop stability are attained.
\end{problem}

\section{Kernelized Velocity State-space Models}\label{sec:Learning_Kernelized_Approximate_Models}
To solve Problem~\ref{Problem:1}, we will utilize the representer theorem to parameterize the nonlinear functions within the velocity state-space model \eqref{eqn:velocity}. Given a finite data set $\{x_i,u_i,y_i\}_{i\in[0,s]}$, consider the kernel-based representations of the unknown functions in \eqref{eqn:velocity}:
\begin{subequations}
\label{eqn:kernel_trick}
\begin{align}
\nabla_x f(\tilde x_k,\tilde u_k) &\approx
    \underbrace{\begin{bmatrix}
         \boldsymbol{\alpha}_{A}^{1,1} & \dots & \boldsymbol{\alpha}_{A}^{1,n} \\
         \vdots & \ddots & \vdots \\
         \boldsymbol{\alpha}_{A}^{n,1} & \dots & \boldsymbol{\alpha}_{A}^{n,n} 
    \end{bmatrix}}_{A_\alpha} \left(I_ n\otimes \Bar{\mathbf{K}}_{xu}(x_k,u_k) \right), \label{eqn:kernel_trick:a}\\
\nabla_u f(\tilde x_k,\tilde u_k) &\approx 
    \underbrace{ \begin{bmatrix}
         \boldsymbol{\alpha}_{B}^{1,1} & \dots & \boldsymbol{\alpha}_{B}^{1,m} \\
         \vdots & \ddots & \vdots \\
         \boldsymbol{\alpha}_{B}^{n,1} & \dots & \boldsymbol{\alpha}_{B}^{n,m} 
    \end{bmatrix}}_{B_\alpha} \left(I_m \otimes \Bar{\mathbf{K}}_{xu}(x_k,u_k) \right), \label{eqn:kernel_trick:b}\\
  \nabla_x h(\tilde x_k)   &\approx
    \underbrace{\begin{bmatrix}
         \boldsymbol{\alpha}_{C}^{1,1} & \dots & \boldsymbol{\alpha}_{C}^{1,n} \\
         \vdots & \ddots & \vdots \\
         \boldsymbol{\alpha}_{C}^{p,1} & \dots & \boldsymbol{\alpha}_{C}^{p,n} 
    \end{bmatrix}}_{C_\alpha} \left(I_n \otimes \Bar{\mathbf{K}}_{x}(x_k)\right), \label{eqn:kernel_trick:c}
\end{align}
\end{subequations}
where $\boldsymbol{\alpha}_{l}^{i,j}:= \col(\alpha_{l,1}^{i,j},\dots,\alpha_{l,s}^{i,j})^\top$ is a row vector of constant coefficients such that $\boldsymbol{\alpha}_{l}^{i,j} \in\Rset^{1\times s}$ for all $l\in\{A,B,C\}$. $\Bar{\mathbf{K}}_{xu}(x_k,u_k):=\col(K_{\col(x_1,u_1)}(\col((x_k,u_k)),\dots,K_{\col(x_s,u_s)}(\col(x_k,u_k)))$ and $\Bar{\mathbf{K}}_{x}(x_k):=\col(K_{x_1}(x_k),\dots,K_{x_s}(x_k))$ are vectors of functions where each element evaluates the kernel function at one of the data points. 
\begin{remark} \label{rem_ker} A more accurate approximation of the velocity dynamics \eqref{eqn:velocity} can be achieved by parameterizing the kernel functions with $x_k$, $x_{k+1}$, $u_k$, and $u_{k+1}$. However, it suffices to parameterize the kernel functions using only $(x_k, u_k, u_{k+1})$, because $x_{k+1}$ is uniquely determined by $(x_k, u_k)$ according to the system dynamics \eqref{eqn:state_space}. Furthermore, since $u_{k+1}$ is computed based on the measured state $x_{k+1}$, the parametrization of the kernel functions can be reduced to ($x_k,u_k$) only.
\end{remark}

The kernelized representations in \eqref{eqn:kernel_trick} can be used to construct the matrices 
\begin{align*}
    \Hat{A}(x_k,u_k) & = \begin{bmatrix}
        I & C_\alpha \left(I_n \otimes \Bar{\mathbf{K}}_{x}(x_k)\right) \\
        0 & A_\alpha \left(I_n\otimes \Bar{\mathbf{K}}_{xu}(x_k,u_k) \right)
    \end{bmatrix}, \\
     \Hat{B}(x_k,u_k) & = \begin{bmatrix}
        0 \\
        B_\alpha \left(I_m \otimes \Bar{\mathbf{K}}_{xu}(x_k,u_k) \right)
    \end{bmatrix}, \\
     \Hat{C}(x_k) & = \begin{bmatrix}
        I & C_\alpha \left(I_n \otimes \Bar{\mathbf{K}}_{x}(x_k)\right)
    \end{bmatrix},
\end{align*}
where $\Hat{A}(x_k,u_k)$, $\Hat{B}(x_k,u_k)$ and $\Hat{C}(x_k)$ are the kernel based approximations of $A(\Tilde{x}_k,\Tilde{u}_k)$, $B(\Tilde{x}_k,\Tilde{u}_k)$ and $C(\Tilde{x}_k)$ in \eqref{eqn:velocity}. 

This yields the kernelized data-driven velocity model
\begin{equation}\label{eqn:model}
    \begin{aligned}
    \Hat{z}_{k+1} &= \Hat{A}(x_k,u_k)\Hat{z}_{k} + \Hat{B}(x_k,u_k)\Delta u_k, \\
    \Hat{y}_{k} &= \Hat{C}(x_k)\Hat{z}_{k},
    \end{aligned}
\end{equation}
where it should ideally hold that $\Hat{z}_{k+1}= z_{k+1}$ and $\hat y_k=y_k$ for $\Hat{z}_{k} = z_{k}$. Based on Remark~\ref{rem_ker}, in theory this is possible if the functions on the left hand side of the equations \eqref{eqn:kernel_trick} belong to the corresponding finite dimensional RKHS. 
\begin{remark}
    The number of columns of the matrices of coefficients $A_\alpha$, $B_\alpha$ and $C_\alpha$ and the number of rows of $\left(I_{n/m} \otimes \Bar{\mathbf{K}}_{xu}(x_k,u_k)\right)$ and $\left(I_n \otimes \Bar{\mathbf{K}}_{x}(x_k)\right)$ scale linearly with the number of data samples. However, the dimensions of $\Hat{A}$, $\Hat{B}$ and $\Hat{C}$ are the same as the dimensions of $A$, $B$ and $C$ in \eqref{eqn:velocity}. This means that the kernelized data-driven velocity model \eqref{eqn:model} yields the same computational complexity when employed in a predictive control scheme, as the original velocity model \eqref{eqn:velocity}.
\end{remark}

In order to learn the unknown matrices of coefficients $A_\alpha$, $B_\alpha$ and $C_\alpha$, we set $\hat{z}_k = z_k$ in \eqref{eqn:model}, which allows us to rewrite the state equation as two separate equations, i.e., 
\begin{align}
    \Delta  \Hat{y}_{k}&=\Hat{y}_k - \Hat{y}_{k-1}= C_\alpha\begin{bmatrix}
\left(I_n \otimes \Bar{\mathbf{K}}_{x}(x_k) \right)\Delta x_k 
\end{bmatrix}, \label{eqn:opt1} \\
\Delta  \Hat{x}_{k+1} & =  \begin{bmatrix} A_\alpha &  B_\alpha \end{bmatrix}\begin{bmatrix}
\left(I_n \otimes \Bar{\mathbf{K}}_{xu}(x_k,u_k)\right) \Delta x_k  \\
\left(I_m \otimes \Bar{\mathbf{K}}_{xu}(x_k,u_k)\right) \Delta u_k
\end{bmatrix}. \label{eqn:opt2}
\end{align}
For brevity, for any $k\in\Nset$ define
\begin{align*}
\mathbf{K}^y_k &:= \begin{bmatrix}
\left(I_n \otimes \Bar{\mathbf{K}}_{x}(x_k) \right)\Delta x_k \\ 
\end{bmatrix}, \\
\mathbf{K}^x_k &:= \begin{bmatrix}
\left(I_n \otimes \Bar{\mathbf{K}}_{xu}(x_k,u_k)\right) \Delta x_k  \\
\left(I_m \otimes \Bar{\mathbf{K}}_{xu}(x_k,u_k)\right) \Delta u_k
\end{bmatrix}.
\end{align*} 
Given a set of noiseless state-input-output data $\{x_i,u_i,y_i\}_{i\in[0,s]}$ and predicted $\{\Hat{x}_i,\Hat{y}_i\}_{i\in[0,s]}$ counterparts corresponding to the kernelized model \eqref{eqn:model}, define 
{\small
\begin{align*}
    \Delta \mathbf{Y}^+ &:= \begin{bmatrix}
             \Delta y_1 & \dots & \Delta y_{s-1}
         \end{bmatrix}, \Delta \Hat{\mathbf{Y}}^+ := \begin{bmatrix}
             \Delta \Hat{y}_1 & \dots & \Delta \Hat{y}_{s-1}
         \end{bmatrix}, \\
         \Delta \mathbf{X}^+ &:= \begin{bmatrix}
             \Delta x_2 & \dots & \Delta x_s
         \end{bmatrix}, \Delta \Hat{\mathbf{X}}^+ := \begin{bmatrix}
             \Delta \Hat{x}_2 & \dots & \Delta \Hat{x}_s
         \end{bmatrix}, \\
          \mathbf{Z}^+ &:= \begin{bmatrix}
             z_2 & \dots & z_s
         \end{bmatrix}, \Hat{\mathbf{Z}}^+ := \begin{bmatrix}
              \Hat{z}_2 & \dots &  \Hat{z}_s
         \end{bmatrix}.
\end{align*}}
We will use these definitions to optimally find approximations for $A_\alpha$, $B_\alpha$ and $C_\alpha$ in the least-squares sense. 
\begin{lemma}
\label{lem_ker}
Suppose that the collected set of noiseless state-input-output data $\{x_i,u_i,y_i\}_{i\in[0,s]}$ and the chosen kernel functions are such that the matrices $\begin{bmatrix}
    \mathbf{K}^y_1 & \dots &  \mathbf{K}^y_{s-1}
\end{bmatrix}$ and $\begin{bmatrix}
    \mathbf{K}^x_1 & \dots &  \mathbf{K}^x_{s-1}
\end{bmatrix}$ have full rank. Then the matrices 
\begin{align}
     C_\alpha^{LS} &:= 
         \Delta \mathbf{Y}^+  \begin{bmatrix}
         \mathbf{K}^y_1 & \dots &  \mathbf{K}^y_{s-1}
     \end{bmatrix}^{\dag} \label{eqn:LS_solution1}, \\
     \begin{bmatrix} A_\alpha^{LS} & B_\alpha^{LS} \end{bmatrix}&:=  \Delta \mathbf{X}^+ \begin{bmatrix}
         \mathbf{K}^x_1 & \dots &  \mathbf{K}^x_{s-1}
     \end{bmatrix}^{\dag}, \label{eqn:LS_solution2}
\end{align}
yield the  $A_\alpha$, $B_\alpha$ and $C_\alpha$ that minimize $\|\Hat{\mathbf{Z}}^+ - \mathbf{Z}^+\|_2^2$. 
\end{lemma}
\begin{proof}
By using \eqref{eqn:opt1} and \eqref{eqn:opt2} we can separate the least squares problem $\min\|\Hat{\mathbf{Z}}^+ - \mathbf{Z}^+\|_2^2$ into two separate least squares problems, i.e.,  $\min\|\Delta\Hat{\mathbf{X}}^+ - \Delta\mathbf{X}^+\|_2^2$ and $\min\|\Delta\Hat{\mathbf{Y}}^+ - \Delta\mathbf{Y}^+\|_2^2$. These least squares problems can be solved separately, which yields
\begin{align}
    &\min \| \begin{bmatrix}
        A_\alpha & B_\alpha
    \end{bmatrix} \begin{bmatrix}
         \mathbf{K}^x_1 & \dots &  \mathbf{K}^x_{s-1}
     \end{bmatrix} - \Delta\mathbf{X}^+ \|_2^2,  \\
    & \min \left\| C_\alpha\begin{bmatrix}
         \mathbf{K}^y_1 & \dots &  \mathbf{K}^y_{s-1}
     \end{bmatrix} - \Delta\mathbf{Y}^+\right\|_2^2.
\end{align}
By exploiting suitable pseudoinverse matrices, yields the least squares solutions $A_\alpha^{LS}$, $B_\alpha^{LS}$ and $C_\alpha^{LS}$, which completes the proof.
\end{proof}

A possible approach to attain the full rank property invoked in Lemma~\ref{lem_ker} is to use universal kernel functions \cite{Micchelli_2006} and distinct data points $\{x_i,u_i,y_i\}$.


\section{Kernelized Offset-Free DPC}\label{sec:stability_feasibility}
We will now define the velocity kernel based data-driven predictive control (vKDPC) problem. To this end we firstly define the prediction matrices which are obtained by iterating the kernelized velocity dynamics in \eqref{eqn:model}, i.e., 
\begin{align}
   &\hat \Psi(\boldsymbol{\rho}_k) := 
\begin{bsmallmatrix}
    \prod_{i=0}^{0} \hat A(\rho_{i|k}) \\
    \prod_{i=0}^{1} \hat A(\rho_{i|k}) \\
    \vdots \\
    \prod_{i=0}^{N-1} \hat A(\rho_{i|k})
\end{bsmallmatrix}, \label{eqn:prediction_model1} \\
&\hat \Gamma(\boldsymbol{\rho}_k) := \nonumber \\
& \begin{bsmallmatrix}
   \hat B(\rho_{0|k})  & 0  & \dots & 0\\
    \prod_{i=1}^{1} \hat A(\rho_{i|k}) \hat B(\rho_{0|k})  & \hat B(\rho_{1|k}) & \dots & 0 \\
    \vdots & \vdots & \ddots  & 0 \\
    \prod_{i=1}^{N-1} \hat A(\rho_{i|k}) \hat B(\rho_{0|k})   &  \prod_{i=2}^{N-1} \hat A(\rho_{i|k}) \hat B(\rho_{1|k}) & \dots & \hat B(\rho_{N-1|k}) \label{eqn:prediction_model2}
\end{bsmallmatrix}.
\end{align} 
Then the vKDPC problem can be defined as follows.
\begin{problem} (Velocity form offset-free kernelized DPC) \label{prob:KDPC}
    \begin{subequations}
\label{eq:k3:1}
\begin{align}
&\min_{\Delta \mathbf{u}_k,\boldsymbol{\rho}_k,\mathbf{z}_{k} }  \quad   l_N(z_{N|k},r) + \sum_{i=0}^{N-1} l(z_{i|k},r,\Delta u_{i|k})    \label{eq:k3:1a} \\  
&\text{subject to: } \nonumber  \\
&\mathbf{z}_k =  \hat\Psi(\boldsymbol{\rho}_k) z_{0|k} + \hat\Gamma(\boldsymbol{\rho}_k) \Delta \mathbf{u}_k \label{eq:k3:1b}\\
 &\rho_{0|k} = \col(x_k,u_{0|k}), \quad z_{0|k} = \col(y_{k-1},\Delta x_k), \label{eq:k3:1d}\\
 &(\mathbf{z}_{k}, \Delta \mathbf{u}_{k}) \in \Zset \times \Delta \Uset, \quad (z_{N|k}-r) \in \mathbb{Z}_T.\label{eq:k3:1e} 
\end{align}
\end{subequations}
\end{problem}

Problem~\ref{prob:KDPC} is a constrained nonlinear optimization problem, since $\hat \Psi(\boldsymbol{\rho}_k)$ and $\hat \Gamma(\boldsymbol{\rho}_k)$ are matrices where the entries in general depend on nonlinear functions of $\boldsymbol{\rho}_k$. However, a computationally much more efficient approach is to employ a sequential QP formulation, as stated next.
\vspace{-1mm}
\begin{algorithm}[H]
\caption{Efficient Velocity Form Kernelized DPC}\label{alg:vKDPC_LPV}
\begin{algorithmic}[1]
\State \textbf{Input:} set-point $r$, prediction horizon $N$, cost matrices $P$, $Q$, $R$, prediction matrices $\Hat{\Psi}$, $\Hat{\Gamma}$, and tolerance $\varepsilon$
\For{$k = 0, 1, \dots$}
    \State Measure or estimate the current state $x_k$
    \State \textbf{Fix parameters:} Set initial $\boldsymbol{\rho}_k^{(1)} = \boldsymbol{\rho}_k^{(0)} = \col(\rho^*_{1|k-1}, \dots, \rho^*_{N-2|k-1}, \rho^*_{N-1|k-1}, \rho^*_{N-1|k-1}$) 
    \State Initialize iteration index $i \gets 0$
    \While{$\|\boldsymbol{\rho}_k^{(i+1)} - \boldsymbol{\rho}_k^{(i)}\| > \varepsilon$ \textbf{or} $i = 0$ }
        \State Increment iteration index $i \gets i + 1$
        \State Solve Problem~\ref{prob:KDPC} with fixed $\boldsymbol{\rho}_k^{(i)}$ to obtain $\Delta \mathbf{u}_k^*$
        \State Construct $\mathbf{u}^*_k$ from $\Delta \mathbf{u}^*_k$ and $u_{k-1}$
        \State Compute $\boldsymbol{\rho}_k^{(i+1)}$ by simulating \eqref{eqn:model} with $\mathbf{u}_k^*$, and $x_k$
    \EndWhile
    \State Apply the control input: $u_k = u_{k-1} + \Delta u_{0|k}^*$ to \eqref{eqn:state_space}
\EndFor
\end{algorithmic}
\end{algorithm}
\vspace{-3mm}
Next, we analyze the properties of the velocity system dynamics \eqref{eqn:velocity} in closed-loop with the control law obtained by solving Problem~\ref{prob:KDPC} online, at each time $k$ in a receding horizon manner. For simplicity of exposition we assume that the kernelized velocity dynamics $\hat A, \hat B, \hat C$ are an exact approximation of the velocity dynamics $A, B, C$, which further yields that the prediction matrices $\hat \Psi$ and $\hat \Gamma$ are an exact approximation of $\Psi$ and $\Gamma$, respectively. If this is not the case, the robust recursive feasibility and stability analysis method from \cite{dejong2024} can be employed, as explained in Remark~\ref{remark:4.7}. Note that establishing terminal cost and set conditions for closed-loop stability and recursive feasibility of velocity NMPC/KDPC is still non-trivial; i.e., in \cite[Chapter 4]{cisneros2021quasi} a more conservative,   output terminal equality constraint was used. 
 
\begin{assumption}\label{assm:TerminalCost_and_Set}
Given a reference $r=\col(y_r,\mathbf{0}_{n\times 1})$, there exists a locally stabilizing state-feedback control law $\Delta u_k:=K(z_k-r)$ (which gives $u_k:=u_{k-1} + K(z_k-r)$), a $P\succ 0$ and a
compact set 
$\Zset_T$ with $0$ in its interior such that for all $(z^*_{N|k}-r)\in\Zset_T$ it holds that:
\begin{subequations}\label{eqn:conditions}
\begin{align*}
&\textbf{A}_{cl}(\Tilde{\rho}_{N-1|k+1}) \Zset_T \subseteq \Zset_T, \quad K \Zset_T \subseteq  \Delta \mathbb{U}, \quad \Zset_T \subseteq  \Zset, \\
&\textbf{A}^\top_{cl}(\Tilde{\rho}_{N-1|k+1}) P \textbf{A}_{cl}(\Tilde{\rho}_{N-1|k+1}) - P \preceq  -Q - K^\top R K, 
\end{align*}
\end{subequations}
where $\Tilde{\rho}_{N-1|k+1}:=\col(x^*_{N|k},u^*_{N-1|k}+K(z^*_{N|k}-r))$ and $\textbf{A}_{cl}(\Tilde{\rho}_{N-1|k+1}) := A(x^*_{N|k},u^*_{N-1|k} + K(z^*_{N|k}-r)) + B(x^*_{N|k},u^*_{N-1|k} + K(z^*_{N|k}-r)) K$.  
\end{assumption}

\begin{theorem}\label{theorem:RecursiveFeasibility}
Suppose that Assumption~\ref{assm:TerminalCost_and_Set} holds. At any time $k\in\Nset$, given that Problem~\ref{prob:KDPC} is feasible for $\rho_{0|k}=\rho^*_{1|k-1}$, $z_{0|k}=z^*_{1|k-1}$, then Problem~\ref{prob:KDPC} is feasible at time $k+1$ for $\rho_{0|k+1}=\rho^*_{1|k}$, $z_{0|k+1}=z^*_{1|k}$. 
\end{theorem}
\begin{proof}
Consider the optimal predicted sequences at 
$k$,
\begin{align*}
    \bu_{k}^* &= \{u^*_{0|k},u^*_{1|k},\dots,u^*_{N-2|k}, u^*_{N-1|k}\}, \\
    \boldsymbol{\rho}_{k}^* &= \{\rho^*_{0|k},\rho^*_{1|k},\dots,\rho^*_{N-2|k}, \rho^*_{N-1|k}\},\\
    \bz^*_{k} &= \{z^*_{1|k},z^*_{2|k},\dots,z^*_{N-1|k}, z^*_{N|k}\}.
\end{align*}
Then, at time $k+1$, construct the sub-optimal sequences:
\begin{align*}
    \Tilde{\bu}_{k+1} &= \{u^*_{1|k},u^*_{2|k},\dots,u^*_{N-1|k}, \Tilde{u}_{N-1|k+1}\}, \\
    \Tilde{\boldsymbol{\rho}}_{k+1} &= \{\rho^*_{1|k},\rho^*_{2|k},\dots,\rho^*_{N-1|k}, \Tilde{\rho}_{N-1|k+1}\}, \\
    \Tilde{\bz}_{k+1} &= \{z^*_{2|k},z^*_{3|k},\dots,z^*_{N|k}, \Tilde{z}_{N|k+1}\}, 
\end{align*}
where $\tilde u_{N-1|k+1}:= u^*_{N-1|k} + K(z^*_{N|k}-r)$ and $\Tilde{\rho}_{N-1|k+1} := \col(x^*_{N|k},u^*_{N-1|k} + K(z^*_{N|k}-r))$ . This results in a sub-optimal, but feasible initial condition  $\rho_{0|k+1}=\rho^*_{1|k}$, $z_{0|k+1}=z^*_{1|k}$, as the shifted sub-optimal sequence $\Tilde{\bu}_{k+1}$ satisfies all the constraints in \eqref{eq:k3:1}, due to Assumption~\ref{assm:TerminalCost_and_Set} and the fact that
{\footnotesize
\begin{align*}
    &\tilde z_{N|k+1} = \\
&= \prod_{i=0}^{N-1} A(\Tilde{\rho}_{i|k}) z^*_{1|k}  \\
& \quad + \begin{bmatrix}
    \prod_{i=1}^{N-1} A(\Tilde{\rho}_{i|k})B(\Tilde{\rho}_{0|k}) & \dots & B(\Tilde{\rho}_{N-1|k})
\end{bmatrix}\Tilde{\mathbf{u}}_{k+1} \\
&= A(\Tilde{\rho}_{N-1|k}) \left( \prod_{i=1}^{N-1} A(\rho^*_{i|k}) z^*_{1|k} +
    \prod_{i=2}^{N-1} A(\rho^*_{i|k})B(\rho^*_{1|k})u^*_{1|k} + \dots 
\right)\\
& \quad + B(\Tilde{\rho}_{N-1}|k)\Tilde{u}_{N-1|k} \\
&= A(\Tilde{\rho}_{N-1|k}) \left( \prod_{i=0}^{N-1} A(\rho^*_{i|k}) z^*_{0|k} +
    \prod_{i=1}^{N-1} A(\rho^*_{i|k})B(\rho^*_{1|k})u^*_{0|k} + \dots 
\right)\\
& \quad + B(\Tilde{\rho}_{N-1}|k)\Tilde{u}_{N-1|k} \\
&= A(\Tilde{\rho}_{N-1|k})z^*_{N|k} + B(\Tilde{\rho}_{N-1|k})\Tilde{u}_{N-1|k} \\
&= \textbf{A}_{cl}(x^*_{N|k},u^*_{N-1|k} + K(z^*_{N|k}-r)) z^*_{N|k} \in \Zset_T.
\end{align*}}

The last statement holds due to Assumption~\ref{assm:TerminalCost_and_Set}, which completes the proof. 
\end{proof}

Next, let $\Zset_f(N)$ denote the set of feasible initial states, i.e., $z_{0|k}$, which is positively invariant by Theorem~\ref{theorem:RecursiveFeasibility} for the closed-loop system \eqref{eqn:velocity}-\eqref{eq:k3:1}. Moreover, we define the optimal value function at time $k$ corresponding to \eqref{eq:k3:1a}, i.e., 
\begin{align*}
V(z_{0|k}^\ast,r)&:=J(\mathbf{z}^*_{k}, r, \Delta \mathbf{u}^*_{k},\boldsymbol{\rho}^*_{k})\\
&=l_N(z_{N|k}^\ast,r) + \sum_{i=0}^{N-1} l(z_{i|k}^\ast,r,\Delta u_{i|k}^\ast).
\end{align*} 
\begin{theorem}\label{thm:as}
    Suppose that Assumption~\ref{assm:TerminalCost_and_Set} holds. Then for all $z^*_{0|k}\in\Zset_f(N)$ it holds that system \eqref{eqn:velocity} in closed-loop system with the vKDPC control law obtained by solving \eqref{eq:k3:1} is Lyapunov asymptotically stable with respect to the equilibrium point $r=\col(y_r,\mathbf{0}_{n\times 1})$.
\end{theorem}
\begin{proof}
    Since $z^*_{0|k}\in\Zset_f(N)$, by Theorem~\ref{theorem:RecursiveFeasibility}, Problem~\ref{prob:KDPC} remains recursively feasible for all $k\geq 1$. Consider the sub-optimal initial state $\Tilde{z}_{0|k+1}=z^*_{1|k}$ and the corresponding state, input and scheduling variable sequences $\Tilde{\mathbf{z}}_{k+1}$, $\Delta \Tilde{\mathbf{u}}_{k+1}$ and $\Tilde{\boldsymbol{\rho}}_{k+1}$. Then by optimality of the cost function it holds that 
    {\footnotesize 
    \begin{align*}
        V&(z^*_{0|k+1},r) - V(z^*_{0|k},r) \\
        & \leq J(\Tilde{\mathbf{z}}_{k+1},r, \Delta \Tilde{\mathbf{u}}_{k+1}, \Tilde{\boldsymbol{\rho}}_{k+1}) - V(z^*_{0|k},r) \\
        &= (z^*_{N|k}-r)^\top(\textbf{A}^\top_{cl}(\Tilde{\rho}_{N-1|k+1}) P\textbf{A}_{cl}(\Tilde{\rho}_{N-1|k+1})-P)(z^*_{N|k}-r) \\
        &+(z^*_{N|k}-r) (Q+K^\top R K)(z^*_{N|k}-r) \\
        &-(z_{0|k}^{\ast}-r)^{\top} Q (z_{0|k}^{\ast}-r) -\Delta u_{0|k}^{\ast\top} R \Delta u_{0|k}^{\ast} \\
        &\leq -(z_{0|k}^{\ast}-r)^{\top} Q (z_{0|k}^{\ast}-r).
    \end{align*}}

Then, by exploiting standard MPC stability analysis techniques, suitable positive definite upper and lower bounds on $V(z_{0|k}^\ast, r)$ can be established, i.e.
$$
\alpha_1(\| z_{0|k}^{\ast}-r \|_2^2) \leq V(z_{0|k}^\ast, r) \leq \alpha_2(\| z_{0|k}^{\ast}-r \|_2^2),
$$
where $\alpha_1(s) = \lambda_{\text{min}}(Q)s$. For the computation of $\alpha_2(s)$ we refer to e.g.\cite{MAYNE2000789}. This further yields the desired property via standard Lyapunov arguments.
\end{proof}
\begin{remark}\label{remark:4.7}
When the kernelized velocity data-driven model \eqref{eqn:model} is not exact, then via the framework of \cite{dejong2024}, robust recursive feasibility can still be guaranteed by allowing $z_{0|k}$ to be a (suitably parameterized) free optimization variable. The asymptotic stability of system \eqref{eqn:velocity} in closed-loop with inexact vKDPC as in \eqref{eq:k3:1} is then translated into input-to-state stability.
\end{remark}


\begin{remark}[On the computation of terminal ingredients]\label{rem:terminal}
Assuming that a desired triplet $(x_r, u_r, y_r)$ is known, finding the terminal weight $P \succ 0$ and gain $K$ amounts to solving 
\begin{align*}
    (A(x_r,u_r)+B(x_r,u_r)K)^\top P &(A(x_r,u_r)+B(x_r,u_r)K)\\
    -P & \preceq -Q -K^\top R  K.
\end{align*}
This can be done by using the kernel based approximate matrices $\Hat{A}(x_r,u_r)$, $\Hat{B}(x_r,u_r)$ and $\Hat{C}(x_r,u_r)$, which are constant matrices after substituting $x_r$ and $u_r$. By pre- and post-multiplying with $P^{-1}$, defining the variables $Y=K P^{-1}$ and $O= P^{-1}$ and applying the Schur complement a standard MPC-LMI is obtained. Once the feedback gain $K$ is computed, a terminal set can then be computed by standard methods for the dyanmics $(\hat A(x_r,u_r)+\hat B(x_r, u_r)K)$ \cite{Kernels_control_Claudio}.
\end{remark}

\section{Illustrative Example}\label{sec:example}
To assess the effectiveness of the developed vKDPC scheme, we consider the discretized nonlinear pendulum model from \cite{dejong2024}, i.e.,
\begin{equation}\label{eqn:example}
    \begin{split}
        \begin{bmatrix}
            x_1(k+1) \\
            x_2(k+1)     
        \end{bmatrix} &= 
        \begin{bmatrix}
            1-\frac{b T_s}{J} & 0 \\
            T_s & 1    
        \end{bmatrix}  \begin{bmatrix}
            x_1(k) \\
            x_2(k)     
        \end{bmatrix} +  \begin{bmatrix}
            \frac{T_s}{J} \\
            0    
        \end{bmatrix} u(k) \\
         &- 
        \begin{bmatrix}
            \frac{M L g T_s}{2 J} \sin(x_2(k)) \\
             0
        \end{bmatrix} + \begin{bmatrix}
            1 \\
            0    
        \end{bmatrix} d(k), \\
        y(k) &= x_2(k),
    \end{split}
\end{equation}
where $u(k)$ and $y(k)$ are the system input torque and
pendulum angle at time instant $k$, while $J = \frac{ML^2}{3}$, $M = 1$kg and $L = 1$m are the moment of inertia, mass and length of the pendulum. Moreover, $g = 9.81$m/s$^2$ is the gravitational acceleration, $b = 0.1$ is the friction coefficient and the sampling time is $T_s = 1/30$s.
\begin{figure}[b!]
  \centering
\includegraphics[scale=0.7,trim=0.0cm 0.0cm 0.0cm 0.0cm,clip]{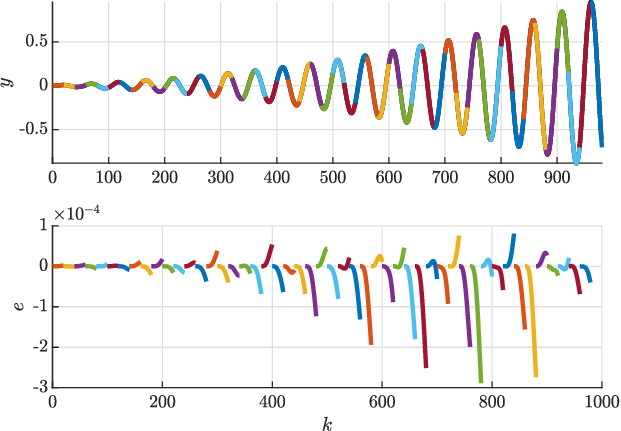}\vspace{-.3cm}
  \caption{$N=20$ step $y$ prediction of kernel model for test data plotted together with the true system trajectory and the error $e:= y-\Hat{y}$.}
  \label{fig:validation}
\end{figure}
In order to learn a kernel based velocity model of the form \eqref{eqn:model} we use the universal \cite{Micchelli_2006} inverse multiquadric kernel function
$k(x_i,x_j) = \frac{1}{\sqrt{1+\frac{1}{\sigma^2}\|x_i-x_j\|^2_2 }}$, with $\sigma = \sqrt{200}$. We perform an open-loop identification experiment where the input signal is generated by dithering a piecewise constant input varying between $\pm 1$ with a multi-sine input constructed with the Matlab function idinput, with the parameters Range [$-0.2, 0.2$], Band [$0, 1$], NumPeriod $1$ and Sine [$25, 40, 1$]. Using the obtained state-input-output data we find the matrices of coefficients $A_\alpha^{LS}$, $B_\alpha^{LS}$ and $C_\alpha^{LS}$ by solving the least squares problems \eqref{eqn:LS_solution1} and \eqref{eqn:LS_solution2} and we construct the multi-step prediction matrices \eqref{eqn:prediction_model1} and \eqref{eqn:prediction_model2}. The total computation time for solving \eqref{eqn:LS_solution1} and \eqref{eqn:LS_solution2} is $2.0064$s for a $2000$ samples dataset. To validate the obtained model we apply a multi-sine input with the same settings to the obtained prediction model and compare the multistep prediction with the test data shown in the first subplot of \figurename{\ref{fig:validation}}. 
\begin{figure}[t!]
  \centering
\includegraphics[scale=0.75,trim=0.0cm 0.0cm 0.0cm 0.0cm,clip]{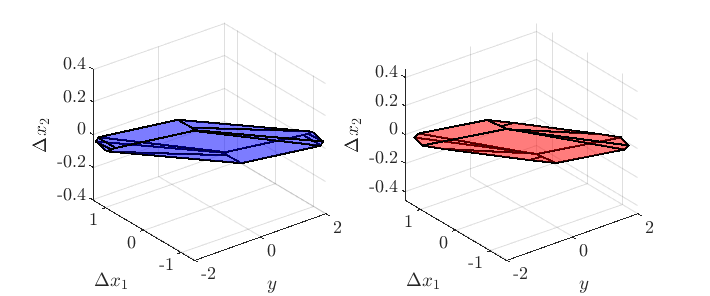}\vspace{-.3cm}
  \caption{Terminal set $\mathbb{Z}_T$ for $r = \col(0.5,0,0)$ obtained using the kernelized velocity model (\textcolor{blue}{blue}) and analytic velocity model (\textcolor{red}{red}).}
  \label{fig:terminal}
\end{figure}
\begin{figure}[b!]
  \centering
\includegraphics[scale=0.8,trim=0.0cm 0.0cm 0.0cm 0.0cm,clip]{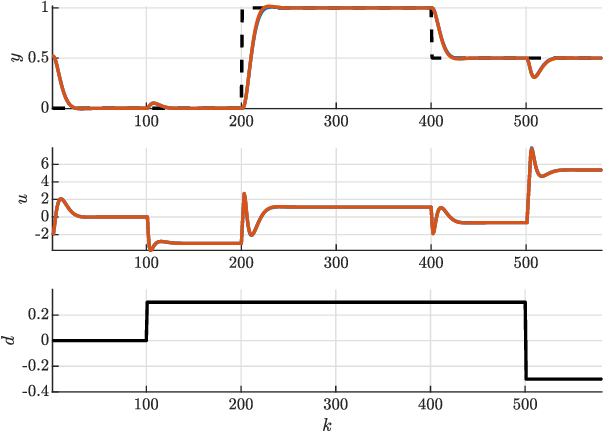}\vspace{-.3cm}
  \caption{vKDPC (\textcolor{Matlab_orange}{\textbf{---}}) vs. vNMPC (\textcolor{Matlab_blue}{\textbf{---}}) trajectories and disturbance signal.}
  \label{fig:trajectories}
\end{figure}
After constructing the multi-step prediction matrices we can directly implement the vKDPC algorithm as formulated in Problem~\ref{prob:KDPC}. 
For the cost function we use $Q = 1000 I_3$ and $R=10$ with a prediction horizon of $N=20$ and tolerance $\varepsilon=1e-8$. The constraint sets are $\mathbb{Z} = \{z : -\col(2,2,2) \leq z \leq \col(2,2,2)\}$ and $\Delta \mathbb{U} = \{\Delta u : -2 \leq \Delta u \leq 2\}$. The terminal cost $P$ is computed as explained in Remark~\ref{rem:terminal} where $x_r = \col(0,y_r)$ and $u_r=\frac{MgL}{2}\sin(y_r)$. 
Since both $x_r$ and $u_r$ depend on the reference signal we have to recompute $P$ when $r$ changes. This amounts to $3$ terminal costs for the example shown in \figurename{\ref{fig:trajectories}}. For every $P$ a corresponding control gain $K$ is computed as explained in Remark~\ref{rem:terminal} and a terminal set $\mathbb{Z}_T$ is computed that fulfills the conditions from Assumption~\ref{assm:TerminalCost_and_Set}. The terminal set is computed using the MPT3 toolbox, see e.g. \cite{dejong2024} for more details about implementation. The terminal set for $r = \col(0.5,0,0)$ is shown in \figurename{\ref{fig:terminal}}.

Next, we consider a piece-wise constant reference $r$ that changes 2 times and a piece-wise constant additive disturbance that is initially zero and then it varies between 2 non-zero values. To validate the performance of the developed vKDPC algorithm we computed the velocity form model \eqref{eqn:velocity} for the pendulum model \eqref{eqn:example} as in \cite[Chapter 4]{cisneros2021quasi} and implemented the corresponding vNMPC algorithm as in Problem~\ref{prob:NMPC}. 

For vNMPC, the analytic velocity model is used to compute $P$ and $\mathbb{Z}_T$  with the same $Q$, $R$, $\mathbb{Z}$ and $\Delta \mathbb{U}$ and $N$ for the matrices $A(x_r,u_r), B(x_r,u_r)$. From \figurename{\ref{fig:terminal}} we observe that the two terminal sets are very similar. We solved both the vNMPC and the vKDPC algorithms using a quasi-LPV formulation and sequential QP. On average the required number of iterations for convergence in Algorithm~\ref{alg:vKDPC_LPV} is equal to $4.0500$ vs $3.2552$ iterations with an average computation time per sampling instant of $0.0683$s vs $0.0438$s for vKDPC and vNMPC respectively, while the offset-free tracking performance of the two algorithms is almost identical. These results demonstrate the effectiveness of the developed data-driven vKDPC algorithm, both in terms of capturing the velocity form nonlinear dynamics and the corresponding offset-free properties, and in terms of efficient online implementation. 
\newpage 
\section{Conclusions}
In this paper we developed a novel kernelized offset-free data-driven predictive control scheme for nonlinear systems. By exploiting the structure of an analytic velocity state-space model, we reduced learning the kernelized velocity model to solving a least squares problem. The resulting offset-free kernelized DPC scheme can be efficiently implemented using sequential QP. We derived terminal cost and terminal set conditions for recursive feasibility and stability of velocity form nonlinear MPC and kernelized DPC for exact kernelized representations. Future work will consider conditions for time-varying references and inexact kernelized velocity models. 


\end{document}